\documentclass[letterpaper, 10 pt, conference]{ieeeconf}  
\usepackage{amsmath}

\DeclareMathOperator*{\argmin}{arg\,min}
\usepackage{amsfonts} 
\usepackage{graphicx}
\usepackage{xcolor}
\usepackage{algorithm}
\usepackage{algpseudocode}

\newtheorem{prop}{Proposition}
\usepackage{comment}
\usepackage{subcaption}

\IEEEoverridecommandlockouts                              
\title{\LARGE \bf Learning Hierarchical Control Systems for Autonomous Systems with Energy Constraints}

\author{Charlott Vallon$^{1}$, Mark Pustilnik$^{1}$, Alessandro Pinto$^{2}$, Francesco Borrelli$^{1}$ 
\thanks{$^{1}$C. Vallon, M. Pustilnik and F. Borrelli are with the Department of Mechanical Engineering,
        University of California, Berkeley,  Berkeley, CA 94720}%
\thanks{$^{2}$A. Pinto is with the NASA Jet Propulsion Laboratory, California Institute of Technology, Pasadena, CA 91011}%
}

\begin{document}

\maketitle
\thispagestyle{empty}
\pagestyle{empty}

\begin{abstract}
This paper focuses on the design of hierarchical control architectures for autonomous systems with energy constraints. We focus on systems where energy storage limitations and slow recharge rates drastically affect the way the autonomous systems are operated. Using examples from space robotics and public transportation, we motivate the need for formally designed learning hierarchical control systems. 
We propose a learning control architecture which  incorporates learning mechanisms at various levels of the control hierarchy to improve performance and resource utilization. 
The proposed hierarchical control scheme relies on high-level energy-aware task planning and assignment, complemented by a low-level predictive control mechanism responsible for the autonomous execution of tasks, including motion control and energy management. 
Simulation examples show the benefits and the limitations of the proposed architecture when learning is used to obtain a more energy-efficient task allocation.

\end{abstract}

\section{Introduction}

This papers focuses on autonomous systems operating in environments where energy management is crucial. Efficient energy use ensures mission completion and prevents situations where energy depletion can lead to system failures or significant operational disruptions. We examine the challenges of designing control systems that can navigate the complexities of energy management in environments where the relationship between the system state and energy consumption is not straightforward.

We motivate our work with two real-world examples in the space-robotic  and public transportation domains (see Fig.~\ref{fig:real-world_scenario}) to demonstrate
the limitations of the current  approach and show how learning could overcome these limitations if properly implemented at each level of the control hierarchy.

\begin{figure}[ht!]
	\centering
\includegraphics[width=0.45\textwidth]{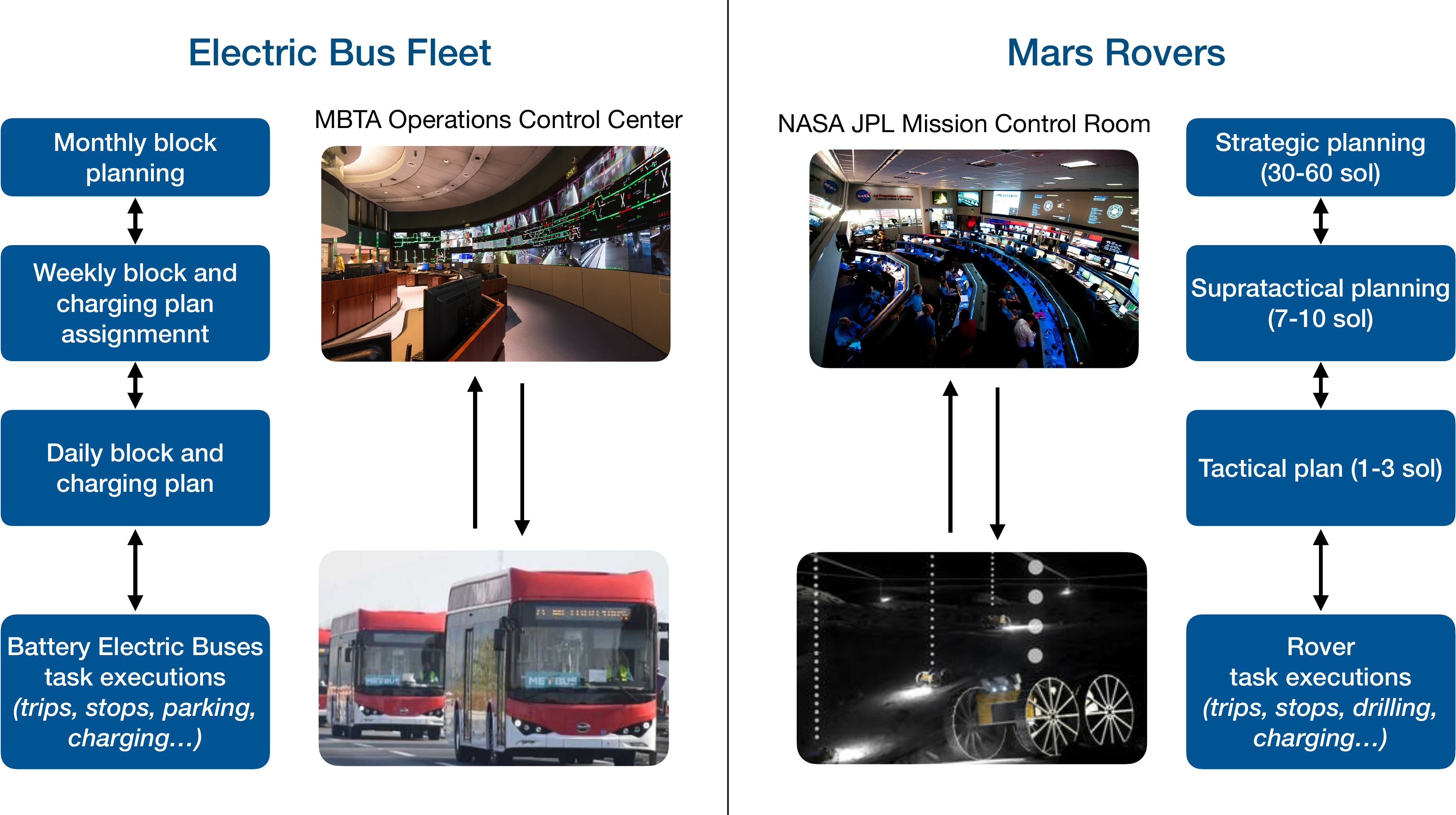}
	\caption{Real-world scenarios for Transit and Space Cyber-Physical Systems.}
	\label{fig:real-world_scenario}
\end{figure}


Perseverance rovers on the surface of Mars have a primary energy source and rechargable batteries to meet peak power demands~\cite{refnasa}. The problem of determining an optimal schedule of scientific activities to be completed by the rovers (e.g. sample collection, surface mapping) is constrained by the state of charge of the batteries and by the available communication windows with the ground operations team on Earth. The planning process is leveled as shown on the right side of Fig.~\ref{fig:real-world_scenario}~\cite{gainesProductivityChallengesMarsa}. The first level generates a strategic plan that typically spans 30-90 Martian days (sols). The second level is concerned with a supratactical plan which typically spans 7-10 sols. The supratactical plan is implemented as a sequence of tactical plans that span 1-3 sols. Planning at each level has to take into account energy consumption and generation, and the time required to execute an activity, which depends on low-level control strategies and environmental conditions (the context). Due to uncertainty, plans are computed using conservative models. It has been shown that the estimated duration of activities is on average 28\%~\cite{gainesProductivityChallengesMarsa,gainesProductivityChallengesMars2016} longer than the actual duration, and that energy requirements are much lower than predicted~\cite{gainesProductivityChallengesMars2016}. 
Thus rovers could be more productive if better models could be learned.



Transit agencies bundle trips into components to be assigned to a certain bus and driver. A typical planning process is shown on the left side of Fig.~\ref{fig:real-world_scenario}, and has the objective of satisfying public demand, maximizing service, minimizing operating cost, and meeting the bus fleet size and operator headcount.
While conventional buses of the same type (e.g. occupancy and geometry) are fully interchangeable between components, the use of Battery Electric Buses (BEBs) limit their assignment to certain components depending on driving range, battery capacity, current state-of-charge (SOC) level, and the varying energy demand of each BEB. 
Prediction of recharging time or the remaining SOC after completing a component are made with simplistic models such as miles per gallon of gasoline-equivalent (MPGe) and miles-per-charge and can be off by as much one order of magnitude~\cite{FBJG2023}. Time and energy associated with a component depend on the implementation and choices of low-level control systems and the context, including traffic, weather, topography, driver behavior, and passenger occupancy.

In both examples, a set of autonomous platforms (referred to as vehicles, agents or robots) is tasked with high-level goals, and a  hierarchical control strategy is used to refine discrete decisions at the higher level into continuous control actions at the lower level. These decisions must take energy consumption and generation into account, but both processes are subject to a high degree of uncertainty.
Inaccurate energy models used by lower level controllers can manifest as erroneous events at the higher level and can lead to a complete mission failure such as BEBs with depleted batteries to be towed to the depot, components of work being canceled, or an entire monthly plan being scrapped, with major disruption to the public service~\cite{FBJG2023}.
In space robotic-missions, mission failures are avoided by overestimating energy requirements, with tremendous impact on mission performance, duration, cost, and science return.

A number of hierarchical control frameworks have been proposed for energy-constrained systems, including for applications ranging from smart grids \cite{stoustrup} to thermal management systems \cite{alleyne2017, sun2023}. 
Here we propose a framework specifically centered around incorporating learning at all levels of the control hierarchy of the energy constrained system to reduce conservatism with enormous  benefits. Learning at the lowest level refers to the ability to use data to update predictive models of energy and safety bounds contextualized on the task type and state of the environment. Learning at the higher levels  refers to the ability of using data  to  reassign tasks, redefine tasks and replan based on new learned models and new tasks. 


In this paper we propose a learning hierarchical control system to address the issue identifies above. The proposed hierarchical control scheme introduces a high level energy-aware task planning and assignment framework for autonomous systems, complemented by a low-level predictive control mechanism responsible for the autonomous execution of tasks, including motion control and energy management. 
The learning-based approach  improves the accuracy of energy consumption models at all levels of the control hierarchy. By leveraging data to refine these models, we reduce the conservatism of energy estimates, leading to better performance and resource utilization.
Simulation examples show the benefits and the limitations of the proposed architecture when learning is used to obtain a more energy  efficient task allocation.



\section{Literature Review }
The ideas discussed in this paper  cover areas of hierarchical planning, multilevel optimization and learning for dynamics and control.
For these topics, the literature is extensive and an exhaustive review goes beyond the scope of this paper.
We limit this discussion to recent relevant work.

\paragraph{Hierarchical planning systems}
\label{par:hierarchical-planning-systems}
The Hierarchical Task Networks (HTNs) \cite{Erol1994UMCPAS} framework models a planning problem with non-primitive tasks, that can be decomposed into a network of tasks by one or more methods, and primitive tasks that cannot be decomposed. Given an HTN domain model and an initial task network as a goal, a planning system~\cite{SchreiberLilotane2021,nau2021gtpyhop} decomposes the goal by applying methods to tasks until only primitive tasks are left in the plan. The development of the domain model task and several approaches have emerged in recent years to learn HTNs from demonstrations and plan traces (see \cite{LearningHTNbyObservationNejati2006,pmlr-v155-chen21d,HTNMakerNDHogg2009,zhuoLearningHierarchicalTask2014,hayesAutonomouslyConstructingHierarchical2016,liLearningProbabilisticHierarchical2014a,hafnerDeepHierarchicalPlanning2022}). Hierarchical reinforcement learning~\cite{Barto2003RecentAI} is another framework based on (Semi) Markov Decision Processes to learn hierarchical policies. Both approaches are monolithic in nature and do not support a heterogeneous and modular representation of the planning problem.

The architectural approach to hierarchical planning features several levels of decision-making agents. A unified framework known as the 4D/RCS architecture is presented in \cite{AlbusTheoryOfIntelligence1991,Albus4DRCS2002}. 
Other hierarchical architectures include the general framework in~\cite{AlamiArchitectureForAutonomy1998}, CLARAty \cite{VolpeClaraty2001}, architectures for self-driving vehicles used in the DARPA Grand Challenge \cite{SpecialIssueDarpaGrandChallenge2006}, and our recent proposal \cite{PintoArchitecture2019} focusing on modularity, openness, and reuse. The architectural approach has several advantages such as the freedom of choosing different representations of the decision problem and the solution method for each agent, and the ability to integrate solutions under control of different stakeholders. \emph{To the best of our knowledge, there is no rigorous framework for the analysis and design of such leveled architectures that considers learning.}

\paragraph{Safe Reinforcement Learning and Iterative Learning Control}

Various efforts have been made towards improving safety and robustness of model-free and model-based learning control \cite{14-gu2023review}. 
One approach is to utilize compositional or hierarchical frameworks where safety is maintained by a master controller or safety filter~\cite{3-ivanov2021compositional}. Here the output of a Reinforcement Learning (RL) policy is ``filtered" by a model-based or model-free pre-trained safety filter to ensure the proposed action cannot lead to unsafe conditions~\cite{11-srinivasan2020learning, 12-thananjeyan2021recovery, 13-bharadhwaj2021conservative}. Commonly proposed model-based safety filters, sometimes referred to as safety critics, include Control Lyapunov Functions which alter that RL outputs to ensure the system descends along a predefined cost function~\cite{4-chow2018lyapunovbased, 5-chow2019lyapunovbased, 6-salamati2020lyapunov}, Control Barrier Functions that correct the RL output to ensure constraint satisfaction~\cite{8-cheng2019endtoend, 9-li2019temporal, 10-emam2022safe}, or both~\cite{23-s2rl2023}.

Other approaches aim to learn RL policies that can operate under disturbances and generalize across tasks. Robust adversarial RL takes a game-theoretic approach which assumes disturbance will actively try to destabilize the system~\cite{19-lütjens2020certified, 24-pinto2017robust}. Domain randomization applies standard RL methods to systems with purposefully perturbed parameters to simulate possible disturbances within an expected range~\cite{20-sadeghi2017cad2rl}. These perturbations can be randomly or strategically chosen to optimize exploration. Most of these methods make no assumptions about dynamics or disturbances, and aim for empirically robust performance rather than theoretical guarantees of stability and safety.

In the control world, Iterative Learning Control (ILC) is a strategy that allows learning from previous iterations to improve its closed-loop tracking performance. ILC has been successfully used in industry, and is behind impressive agile robotics implementations~\cite{hock-auro19}. Recently, we have proposed a new ILC technique, Learning Model Predictive Control (LMPC), where the system model, the terminal set and the terminal cost are constructed  from  measurements  and input data coming from  trajectories of previous iterations~\cite{8039204, rosolia2019learning}. 

In the example presented in this paper we use the LMPC  technique~\cite{8039204,  rosolia2019learning}, where the system model, the terminal set and the terminal cost are constructed  from  measurements  and input data coming from  trajectories of previous iterations. Alternative learning techniques can be used in place of the chosen one as long as state and input constraint satisfaction is guaranteed.

\section{Hierarchical Control Design}\label{sec:hierarchical}
We consider a group of $M$ autonomous battery-powered mobile agents that must perform a repetitive set of tasks at $L$ locations. Without loss of generality, we assume that the period of repetition is one day.

We consider a graphical abstraction of a task space map, with nodes $\mathcal{D} = \{D_0, D_1,\ldots,D_L\}$. A node $D_0$ serves as the depot, where all agents are located at the start of the day. A set of $L$ other nodes $\{D_1,\ldots,D_L\}$ represents locations that agents must travel to in order to complete certain tasks (e.g. ``drill", ``collect", ``measure", ``charge"). A subset of these nodes contain charging stations, where agents may be able to replenish their batteries without returning to the depot.
For each pair $(D_i,D_j)$, with $i,j \in [0,L]$, there exists a set of possible routes $R_{i,j}$ that an agent could use to travel between them.
Note that the time and energy required for an agent to travel any route are complex functions of topological conditions, weather, battery dynamics, and vehicle wear. 

Given a set of tasks to be completed, initial states of charge of each agent, and goal states of charge each agent should have before the start of the next day, the problem studied in this paper is the optimal assignment of tasks to agents and the subsequent optimal execution of the tasks by the agents. Without an understanding of energy consumption for each task, agent utilization will be low. 

Here, we consider a two-level hierarchical learning controller to solve this problem. The upper level (Level 1) produces a plan consisting of task assignment and routing for each agent to complete by the end of the day (see Fig.~\ref{fig:framework}). The lower level (Level 0) implements this plan.\footnote{
We note that in real-world scenarios the number of hierarchical levels can be much higher and include different levels for behavioral planning, path planning, and path following \cite{chaal2020, dinh2020, DIXIT201876}.}
As agents complete tasks and travel along routes, estimates of time and energy requirements along each route 
are improved.
Improved estimates can be leveraged to increase the utilization of each vehicle.

Hierarchy levels contain Situational Awareness (SA) components and Decision Maker (DM) components. SA components receive information about the system and environment states and determine a) whether the plan is being executed as expected, and b) whether and how to update the model being used at the respective hierarchy level. 
Decision Maker components receive information about the system state, the model estimate, and whether the plan is being executed as expected. Based on this information, the component determines whether a new plan has to be created, and what that plan should be. 

Signals passed between components include:
\begin{itemize}
    \item $\hat{m}_{i,j}$ corresponds to the model estimates of agent $j$ used by level $i$. A level of a hierarchy may contain multiple agents and thus multiple models. The form of the model varies between hierarchy levels, but always represents estimates of time and energy requirements. 
    \item $\hat{x}_{i,j}$ is the estimate of the state of agent $j$ used by level $i$ of the hierarchy, based on observations.  
    \item $\pi_{i,j}$ is the plan or goal assigned by hierarchy level $j$ to agent $i$. 
    \item $\hat{\pi}_{i,j}$ describes the status of the plan, including past performance and future projections. Thus each $\hat{\pi}_{i,j}$ contains similar information to ${\pi}_{i,j}$.
    Differences between $\pi$ and $\hat{\pi}$ are what can trigger re-planning in any level. 
    \item $\hat{o}$ serves as the main signal passing observations/information from lower to higher hierarchy levels, with $\hat{o}_{i,j}$ representing information about agent $j$ sent to level $i$ of the hierarchy.
\end{itemize} 

\begin{figure}
    \centering
    \includegraphics[width=0.9\columnwidth]{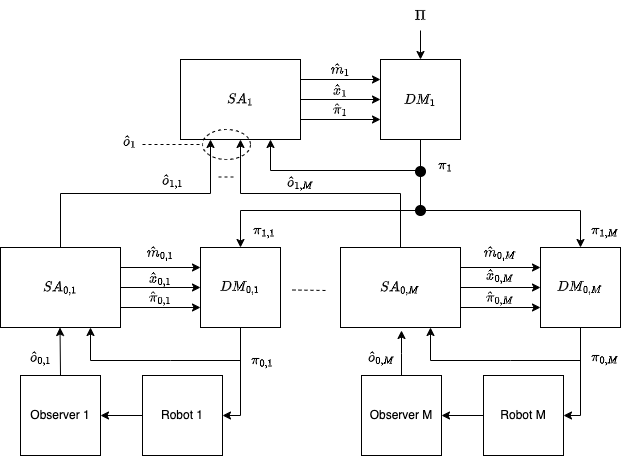}
    \caption{Framework with two levels}
    \label{fig:framework}
\end{figure}

In the remainder of the paper, we refer to the collection of signals corresponding to all agents by dropping the $j$ index, i.e. $\hat{m}_1$ refers to the set of models used to describe agents at Level 1. 

\subsection*{Upper Level: Task Assignment Level}

Level 1 assigns an ordered sequence of tasks for each agent to complete by the end of the day. The relevant states of interest at this level are the location and state of charge of the vehicles, and how much time has passed since arriving at the location. 
At this level, the task assignment problem is treated as a path planning problem over a fully connected graph, where the graph $G(\mathcal{D},\mathcal{R})$ consists of nodes $\mathcal{D}$ and edges $R_{i,j}$ connecting each pair of nodes ($D_i$, $D_j$) with $i,j \in [0,L]$. Each route $R_{i,j}$ is associated with estimates of time and energy expenditures $(\hat{T}_{i,j},\hat{E}_{i,j})$ for that route. These estimates are weights associated with each edge $R_{i,j}$. We refer to the estimates over all routes as $\hat{T}, \hat{E}$.

The upper level Situational Awareness component ($SA_1$) receives observations $\hat{o}_1$ from the lower-level component ($SA_0$), containing information about each vehicle's node location and estimated state of charge.
Based on these measurements, $SA_1$ tracks the execution of the Level 1 plan and updates its respective internal models, including $(\hat{T},\hat{E})$. 
Specifically, this component produces the following signals:
\begin{itemize}
    \item $\hat{x}_1$ is an estimate of the system state at Level 1.
    These state estimates are updated for each vehicle in an event-driven manner, whenever a vehicle reaches a new node $D_i$. At this level, the system state contains information about the position and state of charge (SOC) of each agent and how much time has passed since arriving at the location:
    \begin{align}
        \hat{x}_1(k) &= [\hat{x}_{1,1}(k), \dots, \hat{x}_{1,M}(k)], \\ 
        \hat{x}_{1,i}(k) &= [\hat{P}_i(k), ~\hat{SOC}_i(k), ~\hat{T}_i(k)]
    \end{align}
    where $\hat{P}_i(k)$ estimates the $k$th node visited by the $i$th agent, and $\hat{SOC}_i(k)$ and $\hat{T}_i(k)$ estimate state of charge of $i$th the agent and duration time of event $k$. 

    \item 
    The event-based model $\hat{m}_{1,i}$ predicts the evolution of the upper level state as a function of these estimates:
    \begin{align}
        \hat{x}_{1,i}(k+1) = \hat{m}_{1,i}(\hat{x}_{1,i}(k), u_{1,i}(k),\hat{T}, \hat{E})
    \end{align}
    where the estimates of time and energy expenditures $\hat{T}, \hat{E}$ are based on prior observations $\hat{o}_1$.
    The input $u_{1,i}(k) = [u_{1,i}^r(k), u_{1,i}^c(k)]$ represents the routing and charging input planned for agent $i$ at event $k$. The routing input $u_{1,i}^r(k) \in \{0,1\}^{L \times L}$ tracks the route assigned to agent $i$ at event $k$, where $u_{1,i,a,b}^{r}(k)=1$ indicates the agent $i$ is assigned to travel from node $D_a$ to node $D_b$ at event $k$. If agent $i$ is assigned to charge at event $k$, the charging input $u_{1,i}^c(k) \in \mathbb{R}_{+}$ indicates for how much time the agent is assigned to charge. Note that we must have $u_{1,i}^c(k) = 0$ whenever agent $i$ is not at a charging station node at event $k$.
    
    \item The plan state $\hat{\pi}_1$ describes the progress towards completing the planned tasks. In this example, this consists of the evolution of states $\hat{x}_1$ since the beginning of the task and a forecast of the states until the end of the task:
    \begin{equation}
        \hat{\pi}_1(k) = [\hat{x}_1(0), \dots, \hat{x}_1(k), \bar{\hat{x}}_1(k+1|k), \dots, \bar{\hat{x}}_1(T|k)], \nonumber
    \end{equation} 
    where $\bar{\hat{x}}_1(k+1|k)$ is the prediction made at time step $k$ of the state at time step $k+1$. These predictions are constructed using the estimated Level 1 model $\hat{m}_1$.
\end{itemize}

\textbf{Learning:} As additional observations $\hat{o}_1$ are received by $SA_1$, approximations of the route models $(\hat{T}, \hat{E})$ are improved. 
During repeated execution of a task assignment, the agent keeps track of the past plan state, which records the time and energy required for previously executed tasks. These can be used to update current estimates.
The learning process could be continuous or triggered based on a threshold on the total prediction errors.
Correctly choosing this threshold and how often to update the learned model is not trivial, and incorrect choices can lead to system instability (as demonstrated in the simulation example in Sec.~\ref{sec:simulation}).

The Decision Making component $DM_1$ creates a plan $\pi_1$ containing a sequence of nodes each agent will visit along with the expected time and state of charge at each node:
\begin{align}
    {\pi}_1 &= [{\pi}_{1,1}, {\pi}_{1,2}, \dots, {\pi}_{1,M}], \\ 
    {\pi}_{1,i} &= \begin{bmatrix}
                         D_{ref,1,i} & SOC_{ref,1,i} & k_{ref,1,i}\\
                        D_{ref,2,i} & SOC_{ref,2,i} &  k_{ref,2,i} \\
                        \dots & \dots & \dots 
\end{bmatrix}.\label{eq:planagenti}
\end{align}
The planned delivery route for the $i$th agent (${\pi}_{1,i}$) is an ordered list of nodes the agent is assigned to visit ($D_{ref,j,i}$ is the $j$th node agent $i$ will visit). The plan ${\pi}_{1,i}$ also contains estimates of the time $k_{ref,j,i}$ and state of charge $SOC_{ref, j, i}$ when agent $i$ reaches its $j$th assigned node $D_{ref,j,i}$. 

The plan ${\pi}_1$ corresponds to the task assignment that optimizes a global objective according to the upper level's model estimate $\hat{m}_1$. For example, it may be desirable for agents to jointly complete as many tasks as possible while satisfying time and state of charge constraints; in this case, ${\pi}_1$ could contain the solution of an Energy Constrained Vehicle Routing Problem (ECVRP, see \cite{MONTOYA201787})  based on most recent estimates $\hat{T}, \hat{E}$. This scenario will be explored in Sec.~\ref{example}. 

\textbf{Re-planning:} At this level, re-planning consists of assigning a new route to at least one agent. In our framework this is done only when a agent reaches a task node. Re-planning can be triggered by a divergent state of the plan, or as a result of a significant change in the predictive model. 

\subsection*{Path-Planning/ -Following Level}

Level 0 contains the lower level controllers that convert the plan from Level 1 into actuation commands for each agent. 
Prediction models at this level are different than Level 1. 
At Level 0, each of the $M$ agents is assigned its own Situational Awareness and Decision Making components.

Each Situational Awareness component $SA_{0,i}$ receives observations $\hat{o}_{0,i}$ about the state of the $i$th agent. 
The component produces the following signals:
\begin{itemize}
    \item The signal $\hat{x}_{0,i}(k)$ is an estimate of the system state of agent $i$ at time step $k$, including an estimate of its state of charge, $\hat{SOC}_{0,i}(k)$.
    
    \item A time-based model $\hat{m}_{0,i}$ predicts the evolution of the agent's state with fixed update frequency: 
    \begin{align}\label{eq:lowlevelunicyle}
    \hat{x}_{0,i}(k+1)& = \hat{m}_{0,i}(\hat{x}_{0,i}(k), u_{0,i}(k))
    \end{align}
    subject to the inputs $u_{0,i}(k)$.
    
    \item $\hat{\pi}_{0,i}$ describes the plan state for each agent, consisting of the evolution of the state since the beginning of the task and a forecast of the states until the end of the task:
    \begin{align*}
        \hat{\pi}_{0,i}(k) = & [\hat{x}_{0,i}(0), \hat{x}_{0,i}(1), \dots, \hat{x}_{0,i}(k), \\
        & ~~~~~~~~ \bar{\hat{x}}_{0,i}(k+1|k), \dots, \bar{\hat{x}}_{0,i}(T|k)],
    \end{align*}
    where $k$ is the current time step. The state forecast for agent $i$ is constructed using $\hat{m}_{0,i}$. 
\end{itemize}

\textbf{Learning:} As additional observations $\hat{o}_{0,i}$ are received by the agent's Situational Awareness component, we construct a more accurate system model $\hat{m}_{0,i}$. 
Learning can be triggered based on a threshold on the prediction errors, and agent models can be updated one at a time.

At each time step $k$, the Decision Making component $DM_{0,i}$ selects an actuation input to apply to agent $i$ based on the goals outlined in $\pi_{1,i}$ (\ref{eq:planagenti}) and its current estimated state $\hat{x}_{0,i}(k)$. 
Specifically, at each time step $k$ while agent $i$ is moving from the $p$th reference node $D_{ref,p,i}$ to $D_{ref,p+1,i}$, the lower level plan $\pi_{0,i}(k)$ contains the open-loop sequence of inputs $u_{0,ref,i}(k)$ that should be applied to agent $i$ in order to reach the next reference node prescribed by $\pi_{1,i}$:
\begin{align}
    \pi_{0,i}(k) = u_{0,ref,i}(k),
\end{align}
where $u_{0,ref,i}(k)$
contains the solutions to the optimal control problem:
\begin{subequations}\label{eq:llocgen}
\begin{align}
    {\bf{u}}^{\star}(k),~{\bf{x}}^{\star}(k), ~T^{\star}&(k) =   \nonumber\\
    \argmin_{{\bf{u}}, {\bf{x}}, T}~~ & T \label{oc}\\
    s.t. ~~~ & x_{k|k} = \hat{x}_{0,i}(k) \label{initial_state}\\
    & x_{t+1|k} = \hat{m}_{0,i} (x_{t|k}, u_{t|k}) \label{evolution}\\
    & x_{t|k} \in \mathcal{X} ~~~~~~~\forall t \in 
    \{k, \dots, T\} \label{state_constraints}\\
    & u_{t|k} \in \mathcal{U} ~~~\forall t \in \{k, \dots, T-1\} \label{input_constraints}\\
    & x_{T|k} \in D_{ref,p+1,i} \label{time_goal} \\
    & SOC_{T|k} \geq SOC_{ref,p+1,i} \label{soc_goal}\\
    & T \leq k_{ref, p+1, i}\label{actual_time_goal}\\
    u_{0,ref,i}(k) = {\bf{u}}^{\star}(k) &\nonumber\\
    = [u^{\star}_{k|k}&,~u^{\star}_{k+1|k}, \dots, u^{\star}_{T-1|k}] 
\end{align}    
\end{subequations}
where we use the notation $x_{T|k} \in D_{ref,p+1,i}$ in (\ref{time_goal}) to indicate that the agent is predicted to reach node $D_{ref,p+1,i}$ at time step $T$. This may correspond to a set of geographical constraints as well as potential velocity constraints (e.g. the agents must come to a stop to be considered ``at a node"). 
The optimal control problem plans a state and input trajectory from the agent's current state to the next assigned reference node $D_{ref,p+1,i}$ while satisfying time and energy constraints imposed by $\pi_{1,i}$. The optimal control problem (\ref{eq:llocgen}) minimizes the time for the agent to reach $D_{ref,p+1,i}$ (\ref{time_goal}) without draining the battery below the reference state of charge according to the plan $\pi_{1,i}$ (\ref{soc_goal}). With a slight abuse of notation, constraint (\ref{actual_time_goal}) ensures the planned trajectory arrives at the reference node $D_{ref,p+1,i}$ before the reference time $k_{ref,p+1,i}$.
At time $k$, only the first calculated reference input is applied to the agent:
\begin{align}
    u_{0,i}(k) = u^{\star}_{k|k}.
\end{align}
At time $k+1$ a new $u_{0,ref,i}(k+1)$ is calculated with (\ref{eq:llocgen}).



\textbf{Re-planning:}
At this level, re-planning consists of assigning a new input sequence to be applied to each agent. 
The open-loop solution to (\ref{eq:llocgen}) contains an estimated open-loop state trajectory, which can be evaluated against $\hat{\pi}_{0,i}(k)$ in order to determine whether the plan must be revised. Re-planning can be triggered while the agent is moving between nodes. 
Sec.~\ref{ssec:lowlevel} outlines a computationally efficient way to solve (\ref{eq:llocgen}) that can be implemented in real time.

\section{Future Research}
In the next section we introduce an example that illustrates our proposed hierarchical control systems for energy-constrained autonomous systems. 

Here we want to highlight that Sec.~\ref{sec:hierarchical} introduced a novel  framework which addresses the operational challenges arising from limited energy storage capacities and the constraints of recharge rates in autonomous systems. It establishes the methodological basis for incorporating learning mechanisms at various levels of the control hierarchy, aiming to optimize performance and resource utilization within these constraints.

The principal objective of this paper is to formally define the need of integrating learning mechanisms across different tiers of hierarchical control, contextualized within energy limitations. 
The intent is to augment the efficiency of autonomous systems operations through a control architecture that employs learning for enhanced task allocation and energy management. 
The work here leads to theoretical questions that remain unaddressed within the current scope. These include:
\begin{enumerate}
    \item The definition of criteria for safe and efficient planning within hierarchical control systems.
    \item The elucidation of methods for integrating learning without leading to system infeasibility or instability.
    \item The development of real-time capable approaches that consider the computational constraints of low-level controllers and communication delays.
\end{enumerate}
While some preliminary questions regarding infeasibility are explored in \cite{vallon2024learning2} in the simplified context of a two-layer deterministic hierarchy, this paper utilizes a more realistic example to highlight specific aspects of the problem space at large. We not only demonstrate the applicability of our framework but also emphasize its potential to refine the problem definition and expand the solution space for energy-efficient autonomous system operation.

\section{Example} \label{example}



We consider an electric delivery fleet of $M=2$ vehicles needing to serve a number of customers every day. The vehicles start each day fully charged in a central depot, where they must return at the end of each day after visiting their assigned customers.
The fleet's goal is to serve as many customers as possible while minimizing the total route time (including charging time, if needed), and meeting the following constraints: 
\begin{enumerate}
    \item each vehicle's battery state of charge must remain positive throughout the route, and
    \item the time duration of each route must not exceed a daily time limit. 
\end{enumerate}
We refer to each day during which routes are assigned and executed as an ``iteration". 
We model the set of interest points to be visited as nodes on a fully connected graph $G(D, \mathcal{R})$, with nodes $D$ and edges $\mathcal{R}$.
The nodes comprise of the depot node $D_0$, as well as $20$ customer nodes $\{D_1,...,D_{20}\}$ and two charging stations $\{D_{21},D_{22}\}$. 
The position of each node is chosen randomly with uniform distribution over a 2D grid. All units are normalized. 

The dynamics of each vehicle are described by
\begin{subequations}\label{eq:CTdyn}
    \begin{align}
    & \dot{z}(t) = v(t)  \cos{\theta(t)} \\
    & \dot{y}(t) = v(t)  \sin{\theta(t)} \\
    & \dot{soc}(t) =  -\alpha_{i,j}  v(t) \quad (i,j) \in \mathcal{R} \label{eq:socdelta}\\
    & \dot{v}(t) = a(t) \\
    & \dot{\theta} = \delta(t)
\end{align}
\end{subequations}
where $z(t)$ and $y(t)$ are the horizontal and vertical position on the map at time $t$, $v(t)$ the velocity, $soc(t)$ the vehicle's state of charge, $\theta(t)$ the velocity direction, and $a(t)$ and $\delta(t)$ are the acceleration and steering command. 
The energy consumption rate (\ref{eq:socdelta}) is proportional to the vehicle's velocity with a route-specific value $\alpha_{i,j}$. 
For each route $R_{i,j}$, the true parameter $\alpha_{i,j}$ changes each day (each iteration), and is set at the beginning of each iteration by sampling from a uniform distribution
\begin{align}
    \alpha_{i,j} \sim U(\alpha^{i,j}_{min},\alpha^{i,j}_{max}) \quad \forall (i,j) \in \mathcal{R} \label{alpha}.
\end{align}
The pair ($\alpha^{i,j}_{min},\alpha^{i,j}_{max}$) for each edge is set only once for the whole simulation such that:
\begin{align}
    \alpha_{min} \leq \alpha^{i,j}_{min} < \alpha^{i,j}_{max} \leq \alpha_{max} \quad \forall (i,j) \in \mathcal{R},
\end{align}
where $\alpha_{min} = 0.2$ and $\alpha_{max}=0.5$.

We assume a constant charging rate when the vehicle is at a charging station: $\dot{soc}(t) = 3$.

The vehicles are subject to the state and input constraints 
\begin{align}
    \label{eq:CTcons}
    \mathcal{X} & = \begin{cases}
			-\pi~ [rad] \leq \theta \leq \pi ~[rad] \\
            0  \leq v \leq V_{max}^{i,j} ~[m/s]\quad (i,j) \in \mathcal{R}  \\
            0 \leq soc \leq 100
		    \end{cases} \\
    \mathcal{U} &= \begin{cases}
        -5 ~[m/s^2] \leq a \leq 5 ~[m/s^2] \\
        -0.5\pi~ [rad/s] \leq \delta \leq 0.5\pi ~[rad/s]
    \end{cases}
\end{align}
where the maximal velocity to travel along an edge $(i,j)$ is randomized at the beginning of the simulation and remains constant:
\begin{align}
    V_{max}^{i,j} \sim U(V_{min},V_{max}) \quad \forall (i,j) \in \mathcal{R}
\end{align}
where $V_{min}=3$ and $V_{max}=10$. 

Traveling along each route $R_{i,j}$ requires an associated non-negative time $T_{i,j}$ and energy $E_{i,j}$, which depend in part on the effectiveness of the lower level controller as well as the parameters $\alpha_{i,j}$ and $V^{i,j}_{max}$. 
To prevent situations in which a vehicle is not able to return to the depot by the end of the allotted time, or in which a vehicle's battery runs out in the middle of a tour, the upper level controller initially uses conservative estimates $\hat{T}_{i,j}$ and $\hat{E}_{i,j}$. These estimates may be so conservative that they make it infeasible to plan routes that visit all customers while satisfying battery and time constraints. 
The goal is to utilize the iterative nature of delivery tasks to improve estimates $\hat{T}_{i,j}$ and $\hat{E}_{i,j}$ using data from each iteration of route $R_{i,j}$. 
As more data is collected with each iteration, the time and energy cost estimates improve and the upper level controller is able to plan less conservative routes.

\subsection{Upper Level: Task Assignment}
The upper level controller has two main objectives. 
The first objective is to use collected data from previous iterations to update estimates for time and energy expenditures $\hat{T}_{i,j}, \hat{E}_{i,j}$ along each route $R_{i,j}$. This is performed by the upper level's $SA_1$ component (see Figure~\ref{fig:framework}).
The second objective is to plan routes for the agents that visit as many customers as possible in the daily allotted time, while avoiding vehicles running out of charge. This is performed by the upper level's $DM_1$ component, which sends the planned routes to each agent $i$ via $\pi_{1,i}$ (\ref{eq:planagenti}).


The $SA_1$ component uses observations $\hat{o}_{1}$ to construct estimates of the energy and time expenditures along each route. The $SA_1$ assumes a probability distribution for time $T_{i,j}$ and energy $E_{i,j}$ and learns the bounds of the support of the distribution as data is collected during each iterations. In our experiment, we have used a uniform distribution. The collected data from previous iterations is the total time and energy for each travelled edge $R_{i.j}$. These measured values are then used in the estimation of time and energy models.

To estimate the bounds on time and energy we adapt the algorithm in ~\cite{Monimoy}. The estimation of the upper and lower bound of uniformly distributed variable $X \sim U(a,b)$ is done using $n$ independent measured samples $\{X_i\}_{i=1}^n$:
\begin{subequations}\label{eq:pe}
 \begin{align}
    \hat{a} = \max(a_{min},M - \frac{M-m}{\sqrt[n-1]{1-P_E}}) \\
    \hat{b} = \min(b_{max},m + \frac{M-m}{\sqrt[n-1]{1-P_E}})
\end{align}   
\end{subequations}
where $M$ and $m$ are the maximal and minimal measured values from the $n$ independent samples, and ($a_{min},b_{max}$) are the conservative bounds. Paper ~\cite{Monimoy} present a proof when $a_{min}=-b_{max}$; a proof for a general $(a_{min},b_{max})$ of this estimation algorithm can be found in the Appendix. 
$P_E$ is the confidence interval used in the estimation. This estimation can be used when two or more measurements are available. 
The confidence interval $P_E$ in (\ref{eq:pe}) can be interpreted as a learning rate, where higher $P_E$ values mean the time and energy model estimation takes more samples before converging, but the upper level is likely to produce a feasible plan. 
In Sec.~\ref{sec:simulation}, we compare results from two different confidence intervals, $P_E=0.95$ and $P_E=0.5$.

Once calculated, the estimated time and energy expenditure estimates are used by the $DM_1$ component in the generation of the plan $\pi_{1}$ at the beginning of each day; specifically, the estimated upper bound for the time and energy along each route is used. 
The calculation of plan $\pi_{1}$ is a variation of the Energy Constrained Vehicle Routing Problem (ECVRP), can be expressed as a Mixed Integer Problem (similar to \cite{MONTOYA201787}) and is NP-Hard. To reduce calculation times, the $DM_1$ uses the clustering algorithm for solving the Robust Energy Capacitated Vehicle Routing Problem (RECVRP) described in \cite{ClusterRECVRP}. 

The clustering algorithm calculates a plan for each vehicle in two steps. The first step clusters the customer nodes in $G$ into $M=2$ groups, where each group will be served by a single vehicle. 
The second step calculates the plan $\pi_{1,i}$ for each vehicle $i$ that serves the maximum number of customers possible within the group assigned to the vehicle while minimizing the tour time. This plan is created by solving a vehicle routing problem subject to the estimated time and energy requirements $\hat{T}_{i,j}$ and $\hat{E}_{i,j}$ for each route $R_{i,j}$.
At the beginning of each day a new tour is calculated using the data collected over previous days.

Note that a better estimation of the bounds reduces the probability of generating plans that cannot be executed by the lower level control system within the required time and energy budgets. As more data is collected on each edge, tighter bounds can be calculated, and a less conservative plan can be created (visiting more customers, reducing charging times, reducing total travel times).


\subsection{Lower Level: Path Planning and Following}\label{ssec:lowlevel}
The upper level outputs a plan $\pi_{1,i}$ for each agent $i$ in the form of (\ref{eq:planagenti}), consisting of an ordered sequence of nodes the vehicle will visit, along with the minimum required state of charge and expected reach time for each assigned node.
At each time step $k$, the lower level controller corresponding to agent $i$ converts this plan into actuation commands for the agent $i$ by solving the optimal control problem (\ref{eq:llocgen}), which minimizes the time to reach the next assigned reference nodes while satisfying time and energy constraints put forth by the high-level controller in $\pi_{1,i}$. 

Here we implement a data-driven Model Predictive Control (MPC) controller inspired by \cite{8039204} that approximates the solution to (\ref{eq:llocgen}). 
This approach allows us to take advantage of the iterative nature of the task assignment problem. Each route between two nodes will likely be traveled many times by the agents; data collected from these routes should be utilized to improve estimates about the vehicles' capabilities and improve performance. Each time the agent traverses a particular edge is referred to as a lower level ``iteration". If the lower level controllers can safely complete routes faster and in less time at each iteration, more high-level tasks can be completed.
An additional benefit of this is computational; if $k_{ref,p,i}$ is large for a particular upcoming assigned node $D_{ref,p,i}$, the complexity of solving (\ref{eq:llocgen}) grows, because a longer horizon has to be considered. In contrast, an MPC controller with fixed horizon $N$ has bounded complexity at 
each time step. 

For clarity, we will drop the edge indices in the notation for the remainder of the section, as the proposed methods will be applied in the same manner for each edge. 
We will also drop the agent-specific notation, as the control design is agent-agnostic.

At the lower level, the state of the agent at time $k$ during an iteration $r$ is given by
\begin{align*}
    \hat{x}^r_{0}(k+1) = [z^r_{0}(k), y^r_{0}(k), \theta^r_{0}(k), v^r_{0}(k), SOC^r_{0}(k)],
\end{align*}
where $z^r_{0}(k)$ and $y^r_{0}(k)$ denote the agent's position, $\theta^r_{0}(k)$ its heading angle, $v^r_{0}(k)$ the velocity, and $SOC^r_{0}(k)$ the current state of charge. 
The lower level uses an estimated model 
\begin{align*}
    \hat{x}^r_{0}(k) = \hat{m}^r_{0}(\hat{x}^r_{0}(k), u^r_{0}(k))
\end{align*}
where $\hat{m}^r_{0}$ is an Euler discretization of (\ref{eq:CTdyn}) using a sampling time of $dt = 0.1s$ and an iteration-dependent estimate $\hat{\alpha}^r$ for the energy consumption rate calculated as follows.

At each iteration $r$, the agent's $SA_0$ component utilizes collected observations $\hat{o}_0$ to update its estimate of $\hat{\alpha}^r$. Because the true energy consumption rate $\alpha$ is stochastic (\ref{alpha}), the lower level controller utilizes a conservative estimate. 
This conservative estimate is calculated in two steps.
First, using collected signals ${\bf{v}}_{0}^{r-1}$, ${\bf{soc}}_{0}^{r-1}$ corresponding to the agent's measured velocity and state of charge trajectories during the previous iteration $(r-1)$ of traversing the edge, the $SA$ component estimates $\alpha^r$ using least-squares regression on (\ref{eq:socdelta}). Thus $\alpha^r$ estimates the agent's energy consumption rate at iteration $r-1$. 
The lower layer's conservative estimate for $\hat{\alpha}^r$ is then set as
\begin{align}\label{eq:alphaestimate}
    \hat{\alpha}^r = \max \{ \hat{\alpha}^{r-1}, \alpha^r\}.
\end{align}
At each iteration $r$, the lower level model uses an energy consumption rate estimate corresponding to the highest energy consumption rate seen in a previous iteration. 

The lower level $DM_0$ component converts the higher level plan $\pi_{1}$ into a lower level plan $\pi_{0}$ and actuation inputs to apply to the agent. 
At time step $k_p$, when the agent has reached its $p$th reference node $D_{ref,p}$, the  $DM_0$ component solves an optimal control problem as in (\ref{eq:llocgen}), using ${\hat{m}}_0^r$. The resulting optimal trajectory ${\bf{x}}^{\star}$ is used at iteration $r$ to initialize a set $\mathcal{SS}^r$ between nodes $D_{ref,p}$ and $D_{ref,p+1}$:
\begin{align}\label{eq:batch}
    \mathcal{SS}^{r} = {\bf{x}}^{\star}(k_p) \text{ from }(\ref{eq:llocgen}) \text{ using }(\ref{eq:alphaestimate}).
\end{align}
The set $\mathcal{SS}^{r}$ is only calculated \textit{once} per edge traversal, when the agent is at node $D_{ref,p}$. 
We introduce the cost function
\begin{align*}
    h(\hat{x}_{0}, D) = \begin{cases}
        0 & \hat{x}_0 \in D \\
        1 & \text{else}
    \end{cases}
\end{align*}
which assigns a cost of $1$ to a state $\hat{x}_0$ if it is not in the node $D$, and a cost of $0$ otherwise. Here, we specifically define $\hat{x}_0 \in D$ to indicate that the agent has reached the geographical location corresponding to node $D$ and come to a full stop ($v_0 = 0$). We use this cost function to 
define an iteration-dependent value function over $\mathcal{SS}^{r}$ as
\begin{align}
    V^{r}(\hat{x}_0) = \begin{cases}
        \sum_{t=k}^{T^r} h(x^{\star}_{t}, D_{ref,p+1}) & \text{if } \hat{x}_0 = x^{\star}_{k}, ~x^{\star}_{k}\in \mathcal{SS}^{r} \\
        \infty & \text{else},
    \end{cases}
\end{align}
which associates each state in $\mathcal{SS}^{r}$ with the number of remaining time steps required to reach $D_{ref,p+1}$ according to (\ref{eq:batch}).
The set $\mathcal{SS}^r$ and function $V^r$ are updated at each iteration as new trajectory data becomes available.

At each time step $k$ of the $r$th traversal along an edge, the agent's $DM_0$ component calculates the plan $\pi_{0}(k)$ by solving an optimal control problem that approximates (\ref{eq:llocgen}) by searching for a trajectory of length $N$ from the current state $\hat{x}_0^r(k)$ minimizing the sum of a stage cost and terminal cost:
\begin{subequations}\label{eq:lmpc}
\begin{align}
{\bf{u}}^{r,\star}(k) = \argmin_{{\bf{u}}} & \sum_{t=k}^{N-1} h(x_{t|t}) + V^{r}(\bar{x}) \\
    s.t.~~ & x_{k|k} = \hat{x}_{0}^r(k) \nonumber\\
    & x_{t+1|k} = \hat{m}^r_{0} (x_{t|k}, u_{t|k}) \nonumber\\
    & x_{t|k} \in \mathcal{X} ~~~~~~~~~ \forall t \in \{k, \dots, k+N\} \nonumber\\
    & u_{t|k} \in \mathcal{U} ~~~~ \forall t \in \{k, \dots, k+ N-1\} \nonumber\\
    & \bar{x} = [z_{k+N|k}, y_{k+N|k}, \theta_{k+N|k}, \nonumber\\
    &  ~~~~~~~~~ ~~~~~~~~~~~~~v_{k+N|k}, \bar{SOC}]  \label{eq:lmpcss}\\
    & \bar{x} \in \mathcal{SS}^{r} \label{eq:ss2}\\
    & SOC_{k+N|k}\geq \bar{SOC} \label{eq:lmpcss_soc}\\
    u^r_{0}(k) = u^{r,\star}_{k|k} ~~~~~&
\end{align}
\end{subequations}
The controller (\ref{eq:lmpc}) searches at each time step $k$ for an $N$-step trajectory beginning at the current state estimate $\hat{x}_{0}^r(k)$ and ending in a state in $\mathcal{SS}^{r}$, enforced by constraints (\ref{eq:lmpcss}), (\ref{eq:ss2}). 
This constraint ensures that the planned trajectory ends in a state from which the agent is expected to be able to reach the next node while satisfying state and input constraints, by following the remainder of the reference plan (\ref{eq:batch}) stored in $\mathcal{SS}^{r}$.
Note that the state-of-charge constraint is imposed as an inequality (\ref{eq:lmpcss_soc}); this allows the controller (\ref{eq:lmpc}) to plan trajectories that ensure the agent will reach the next task node in the time and with \textit{at least} the remaining state of charge expected by the high-level task assignment planner, while satisfying all other state constraints. 
At each time step $k$, only the first calculated input is applied, $u^r_{0}(k) = u^{\star}_{t|t}$, and a new control input is calculated using (\ref{eq:lmpc}) at time step $k+1$. 
A formal description of the feasibility, cost-improvement, and convergence properties of this low-level controller (\ref{eq:lmpc}) can be found in \cite{8039204}. 

The plan $\pi_{0}(k)$ concatenates the actuation sequence corresponding to the $N$-step plan ending in $\bar{x} \in \mathcal{SS}^{r}$ with the actuation sequence calculated in (\ref{eq:batch}) beyond $\bar{x}$, denoted ${\bf{u}}^{\star}_{\bar{k}:}$. 
\begin{align}
        \pi_{0}(k) = u_{0, ref}(k) = [{\bf{u}}^{r,\star}&(k), {\bf{u}}^{\star}_{\bar{k}:} ]\label{eq:lowlevelpi0}
\end{align}

We note that the low-level control approach outlined here requires solving a single batch problem (\ref{eq:batch}) at the beginning of each edge traversal. While the agent is traversing the edge, a short-horizon MPC (\ref{eq:lmpc}) is calculated at each time step, reducing computational complexity as compared with solving (\ref{eq:llocgen}) at each time step. 
Additionally, this iterative data-driven controller formulation introduces learning at the lower level. The approach allows us to explore faster routes between two nodes than the batch reference without sacrificing feasibility of the higher level plan. These improvements can be utilized by the higher level controller plan a more efficient task assignment.

\section{Simulation Results}\label{sec:simulation}

This section presents the results for the example problem and hierarchical controller described in Sec.~\ref{example}. 
A Python simulation was created to implement the problem.
At the end of each iteration $r$ (day $r$), the data collected by the lower level from implementing plan $\pi_1$ for day $r$ is sent to the upper level to refine the problem parameters and calculate a new plan for day $r+1$. 

\begin{figure}[htp]\label{fig:Trajs}
\begin{subfigure}{\textwidth}
\includegraphics[width=0.45\textwidth]{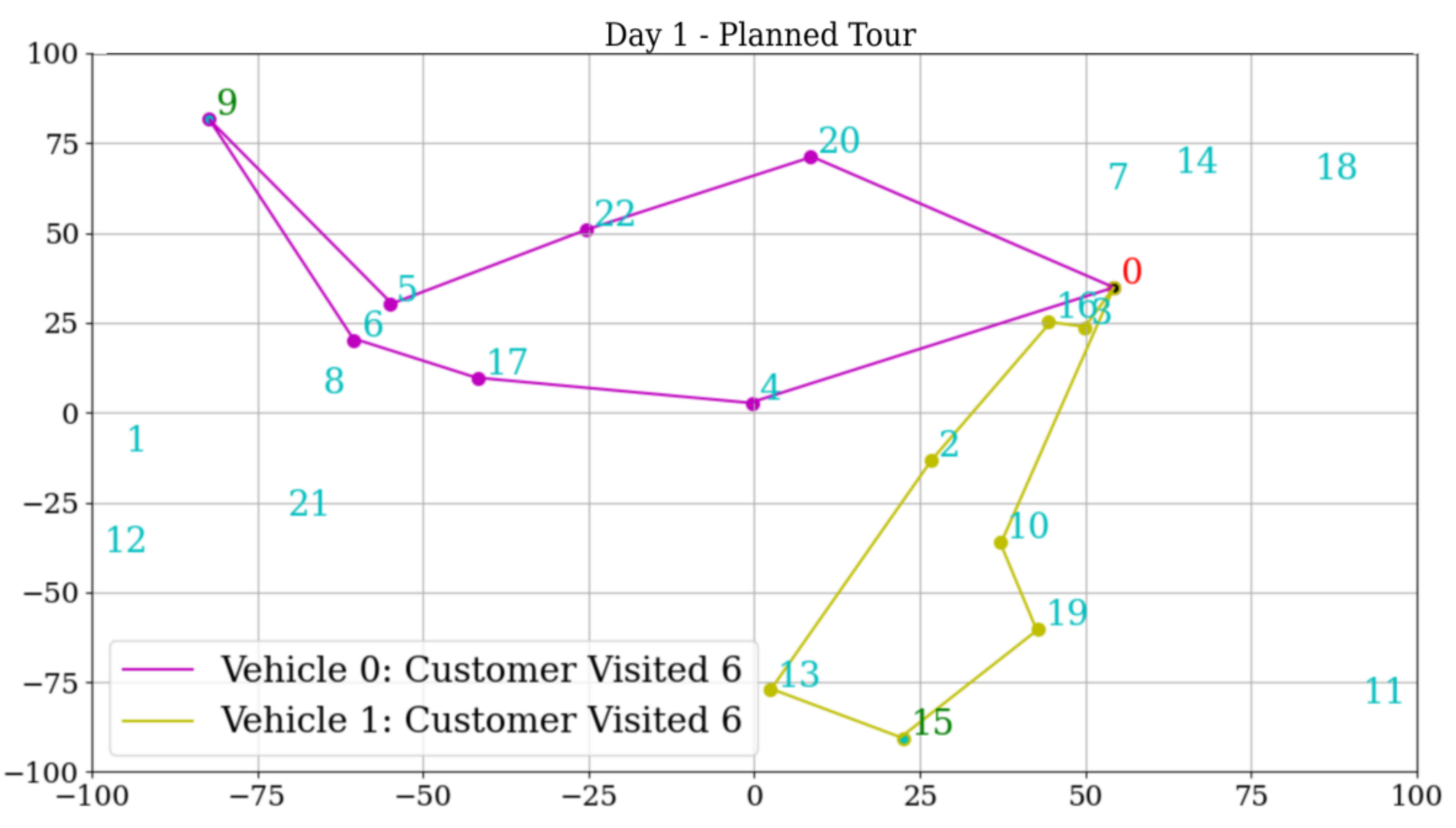}
\end{subfigure}
\bigskip
\begin{subfigure}{\textwidth}
\includegraphics[width=0.45\textwidth]{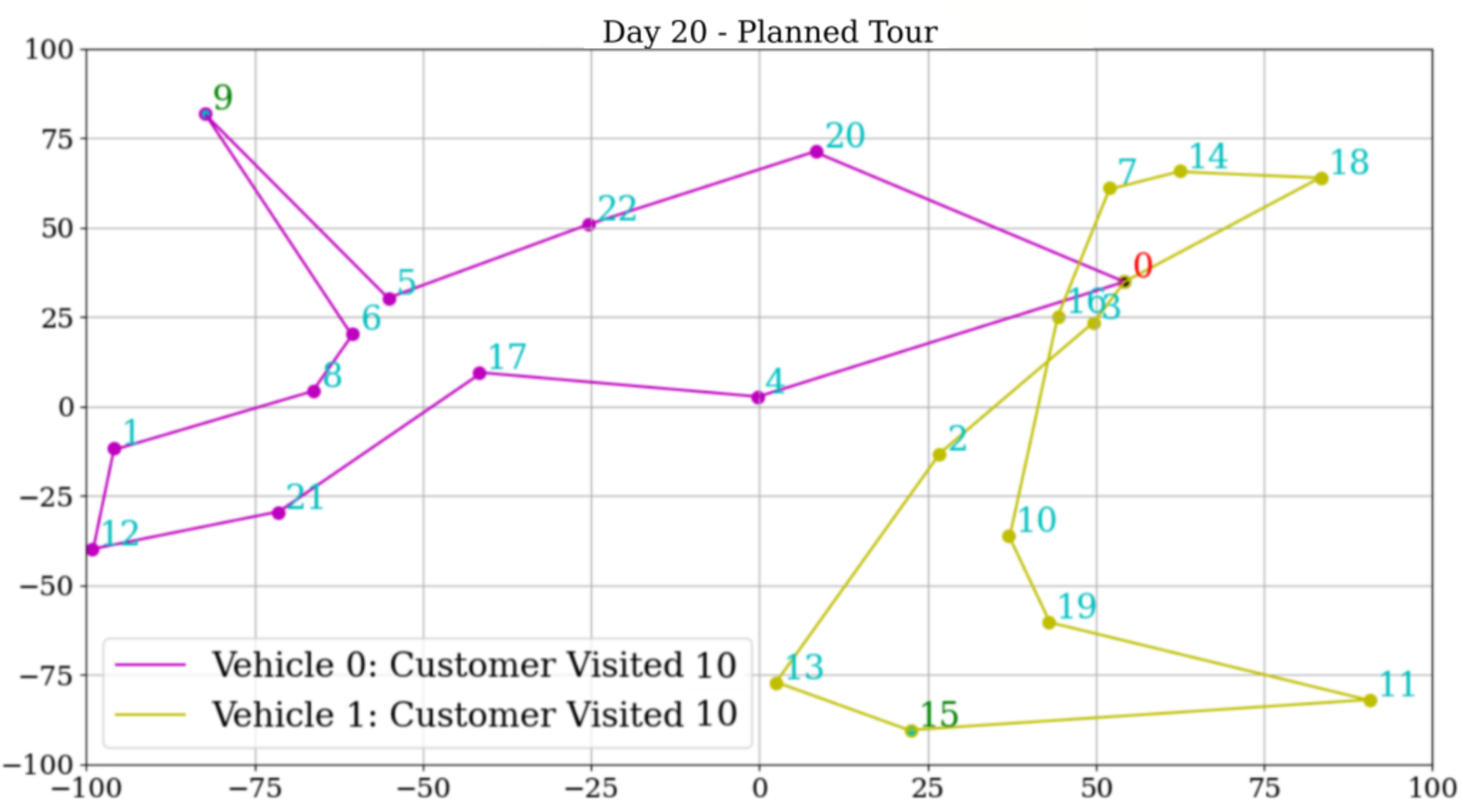}
\end{subfigure}
\caption{Planned tour on day one vs. day 20. The red node is the depot, blue nodes are customers, and green nodes are charging stations.}\label{fig:Trajs}
\end{figure}

\textbf{Improved Upper Level Task Assignment: }
Figure~\ref{fig:Trajs} shows the routes planned on day one and day 20 by an upper level controller using confidence interval $P_E=0.95$. 
On day one (top image), because of initial conservative assumptions made on the time and energy costs associated with each route, the upper level controller is unable to plan a route that includes all customers. Each vehicle visits only six customers. 
However, as vehicle data is collected during more task iterations, the time and energy cost estimates are improved and lowered, and an increasing number of customers can be visited in the allotted time and without running out of charge. 
This increase results from improvements at each level of the hierarchy: as the lower level finds faster and more energy-efficient routes between nodes, the upper level's estimates of time and energy expenditures along each route decreases, allowing it to safely plan tours that visit more customers.

\begin{figure}[ht!]
	\centering
\includegraphics[width=0.45\textwidth]{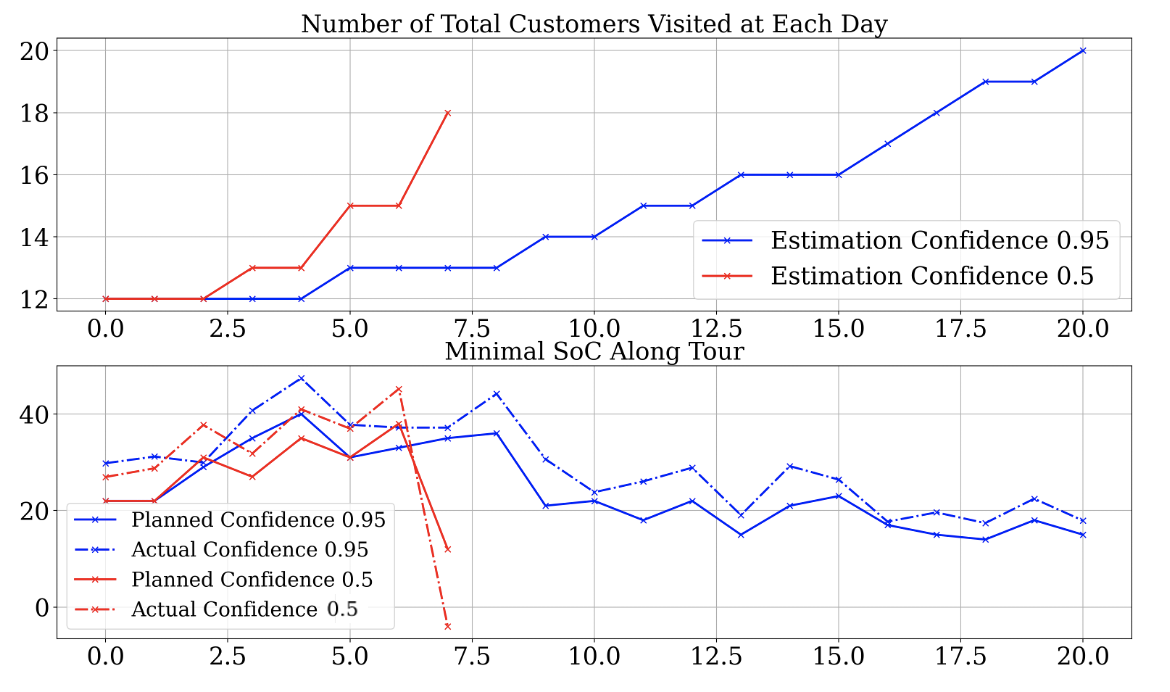}
	\caption{Planned vs. actual energy trajectory along tours. Tours planned with $P_E = 0.5$ initially quickly increase the number of visited customers, but become infeasible on day seven.}
	\label{fig:energy_plan}
\end{figure}

\textbf{Importance of Learning Rate Selection: }
The upper level controller uses the confidence interval $P_E$ (\ref{eq:pe}) as a learning rate that dictates how much weight to give newly collected trajectory data when updating the parameters $\hat{\alpha}_{min}^{i,j}$, $\hat{\alpha}_{max}^{i,j}$, and $V^{i,j}_{max}$. 
Setting the learning rate $P_E$ incorrectly can result in estimates being updated too optimistically, and $DM_1$ creating a plan $\pi_1$ that is infeasible for the vehicle to track. 
This is demonstrated in Fig.~\ref{fig:energy_plan}. 
Figure~\ref{fig:energy_plan} compares the expected and actual energy consumption of a particular vehicle across multiple days, when different learning rates $P_E$ are used. 
The top figure plots how many customers the upper level planned for the agent to visit on each day, while the lower figure plots the lowest expected (and actual) state of charge on each day. 
Blue lines correspond to the case where routes were planned using $P_E=0.95$; in these cases, the energy plan is robust enough so that the vehicle's actual state of charge never dips below the upper level plan's expectation. As a result, the planned tour can be executed by the vehicle.
The number of nodes visited increases gradually as more data is being collected and the energy model is slowly refined. 
In the case where the lower confidence parameter $P_E=0.5$ is used to update the estimates for $\hat{T}_{i,j}$ and $\hat{E}_{i,j}$, the learning occurs more quickly; the number of visited customers initially increases very quickly over the first seven days. 
However, the optimistic update rate eventually leads the upper level to produce a plan on day eight that is not actually feasible, and the vehicle unexpectedly runs out of charge in the middle of the planned tour, causing an iteration failure.

\begin{figure}[t!]
    \centering
    \includegraphics[width=0.45\textwidth]{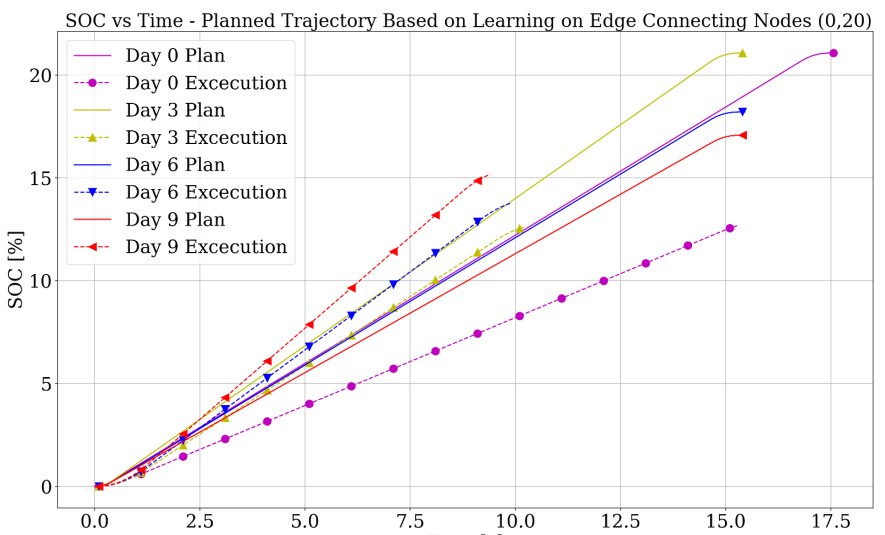}
    \caption{The planned (solid lines) and executed (dashed lines) time vs. state of charge trajectories along the edge (0, 20) over nine days.}
    \label{fig:lmpcfigs}
\end{figure}

\textbf{Robust Lower Level Path Planning/Following: }
Figure~\ref{fig:lmpcfigs} depicts the lower level's planned trajectories on the route between node $D_0$ and node $D_{20}$, for $P_E = 0.95$. As more data is collected and the upper bound on the time and energy expenditure estimates is reduced, the lower level plans routes 
each day that take less energy and less time than the previous day. This is depicted in the solid lines. 

To maintain closed-loop feasibility while tracking the planned routes despite randomness in the true energy depletion rate $\alpha_{0,20}$, the lower level controller 
uses a conservative estimate in its vehicle model.
This ensures that the time and energy costs of the lower level controller 
in closed-loop with the agent experiencing the true energy depletion rate (\ref{eq:socdelta}) is always lower than those of the planned routes. 




\section{Conclusions}
This paper introduces a modular two-level hierarchical control scheme for energy-constrained autonomous systems, where learning is incorporated at each level of the controller.
We demonstrate the necessity for such a framework in applications where energy storage limitations and slow recharge rates hinder the operational efficiency of such systems, thereby limiting their utilization.
Rather than relying on conservative models to guarantee safety, our framework proposes a data-driven approach, where collected data is used to learn different energy models at each level of the control hierarchy.
In a simulation example, we explore the benefits of our modular framework for efficient electric vehicle routing. 
As shown in the example, integration of learning dynamics must be designed in a way that ensures the effects on other levels of the hierarchical scheme will not result in instability or infeasibility of the system. 
Future work will aim to develop theory regarding the selection of a learning rate to ensure stability and feasibility, and to calculate guarantees on how much initial conservatism can be reduced throughout the learning process.

\section{Appendix}
\label{est_proof}
A proof of the estimation algorithm used in the $SA_1$ component of the hierarchical controller is presented. The energy consumption rate for each edge is a random variable uniformly distributed with unknown bounds $\alpha \sim U(a,b)$, where conservative bounds $(a_{min},b_{max})$ are known:
\begin{align}
    a_{min} < a < b < b_{max}
\end{align}
the goal is to estimate the bounds $(a,b)$ from $n$ independent measurements  $\{\alpha_i\}_{i=1}^n$, with some confidence interval of $0<P_\alpha<1$.
\begin{prop}
The lower and upper bounds of a uniformly distributed variable $X \sim U(a,b)$ can be estimate with confidence $1-\alpha$ from $n$ independent identically distributed (i.i.d) samples by:
\begin{align}
    \hat{a} = M - \frac{M-m}{\sqrt[n-1]{\alpha}} \\
    \hat{b} = m + \frac{M-m}{\sqrt[n-1]{\alpha}}
\end{align}
where $M$ and $m$ are the maximal and minimal valued samples out of the $n$ independent samples.
\end{prop}
\begin{proof}
This proof is based on \cite{Monimoy}. We will proof the estimation for the upper bound - the proof for the lower bound is similar. The samples are sorted $X_1\leq X_2\leq...\leq X_n$ from the lowest to the highest value samples. The proof is divides into 3 parts. First, we assume $a=0$ and is known. This implies that for some $c \in (0,1)$ and the maximal sample $X_n$:
\begin{align}
    \mathcal{P}(\frac{X_n}{b} \leq c) = \mathcal{P}(\bigcap_{i=1,..,n}\frac{X_i}{b}<c) = \prod_{i=1}^{n} \mathcal{P}(\frac{X_i}{b} \leq c) = c^n
\end{align}
where, the last equality is true because $\frac{X_i}{b} \sim U(0,1)$. Setting $c = \sqrt[n]{\alpha}$:
\begin{align}
    \mathcal{P}(\frac{X_n}{b} \leq \sqrt[n]{\alpha}) = \alpha
\end{align}
Therefore, we get:
\begin{align}
    \mathcal{P}(\sqrt[n]{\alpha} \leq \frac{X_n}{b} \leq 1) = 1-\alpha
\end{align}
which gives us:
\begin{align}
    \mathcal{P}(X_n \leq b \leq \frac{X_n}{\sqrt[n]{\alpha}}) = 1-\alpha
\end{align}
meaning, if $a=0$ and known, the $1-\alpha$ confidence estimation for $b$ is:
\begin{align}
    \hat{b} = \frac{X_n}{\sqrt[n]{\alpha}}
\end{align}
In the second part of the proof, we allow $a$ to take any known value $a<b$. This implies that the estimation for the upper bound $b$ is:
\begin{align}
    \hat{b} = a + \frac{X_n-a}{\sqrt[n]{\alpha}}
\end{align}
In the third and final part, there isn't any assumption on the unknown lower bound $a$. We claim that conditioning the n-1 largest samples on $X_1$ gives a uniform distribution with the following bounds:
\begin{align}
    X_i \stackrel {i.i.d}\sim U(a,b) \Rightarrow X_i(x|X_1) \stackrel {i.i.d}\sim U(X_1,b)
\end{align}
now, that we have $n-1$ independent uniformly distributed samples with known lower bound, we can estimate the upper bound by:
\begin{align}
    \hat{b} = X_1 + \frac{X_n-X_1}{\sqrt[n-1]{\alpha}}
\end{align}
This concludes the proof.

\end{proof}

\section*{ACKNOWLEDGMENT}
Part of this research was carried out at the Jet Propulsion Laboratory, California Institute of Technology, under a contract with the National Aeronautics and Space Administration (80NM0018D0004).

\bibliographystyle{IEEEtran}
\bibliography{IEEEfull,bib}

\begin{thebibliography}{10}
\providecommand{\url}[1]{#1}
\csname url@samestyle\endcsname
\providecommand{\newblock}{\relax}
\providecommand{\bibinfo}[2]{#2}
\providecommand{\BIBentrySTDinterwordspacing}{\spaceskip=0pt\relax}
\providecommand{\BIBentryALTinterwordstretchfactor}{4}
\providecommand{\BIBentryALTinterwordspacing}{\spaceskip=\fontdimen2\font plus
\BIBentryALTinterwordstretchfactor\fontdimen3\font minus \fontdimen4\font\relax}
\providecommand{\BIBforeignlanguage}[2]{{%
\expandafter\ifx\csname l@#1\endcsname\relax
\typeout{** WARNING: IEEEtran.bst: No hyphenation pattern has been}%
\typeout{** loaded for the language `#1'. Using the pattern for}%
\typeout{** the default language instead.}%
\else
\language=\csname l@#1\endcsname
\fi
#2}}
\providecommand{\BIBdecl}{\relax}
\BIBdecl

\bibitem{refnasa}
https://mars.nasa.gov/mars2020/spacecraft/rover/electrical power/, 2022.

\bibitem{gainesProductivityChallengesMarsa}
D.~Gaines, G.~Doran, H.~Justice, G.~Rabideau, S.~Schaffer, V.~Verma, K.~Wagstaff, A.~Vasavada, W.~Huffman, R.~Anderson \emph{et~al.}, ``{Productivity Challenges for Mars Rover Operations: A case study of Mars science laboratory operations},'' \emph{Jet Propulsion Lab. TR D-97908}, 2016.

\bibitem{gainesProductivityChallengesMars2016}
D.~Gaines, R.~Anderson, G.~Doran, W.~Huffman, H.~Justice, R.~Mackey, G.~Rabideau, A.~Vasavada, V.~Verma, T.~Estlin \emph{et~al.}, ``Productivity challenges for mars rover operations,'' in \emph{Proceedings of 4th Workshop on Planning and Robotics (PlanRob)}.\hskip 1em plus 0.5em minus 0.4em\relax London, UK, 2016, pp. 115--125.

\bibitem{FBJG2023}
J.~Guanetti, Y.~Kim, X.~Shen, J.~Donham, S.~Alexander, B.~Wootton, and F.~Borrelli, ``Increasing electric vehicles utilization in transit fleets using learning, predictions, optimization, and automation,'' in \emph{2023 IEEE Intelligent Vehicles Symposium (IV)}.\hskip 1em plus 0.5em minus 0.4em\relax IEEE, 2023, pp. 1--6.

\bibitem{stoustrup}
J.~Bendtsen, K.~Trangbaek, and J.~Stoustrup, ``Hierarchical model predictive control for resource distribution,'' in \emph{49th IEEE Conference on Decision and Control (CDC)}, 2010, pp. 2468--2473.

\bibitem{alleyne2017}
H.~C. Pangborn, M.~A. Williams, J.~P. Koeln, and A.~G. Alleyne, ``Graph-based hierarchical control of thermal-fluid power flow systems,'' in \emph{2017 American Control Conference (ACC)}, 2017, pp. 2099--2105.

\bibitem{sun2023}
Q.~Hu, M.~R. Amini, A.~Wiese, R.~Semel, J.~B. Seeds, I.~Kolmanovsky, and J.~Sun, ``Robust thermal management of electric vehicles using model predictive control with adaptive optimization horizon and location-dependent constraint handling strategies,'' \emph{IEEE Transactions on Control Systems Technology}, vol.~31, no.~5, pp. 2119--2131, 2023.

\bibitem{Erol1994UMCPAS}
K.~Erol, J.~A. Hendler, and D.~S. Nau, ``{UMCP: A Sound and Complete Procedure for Hierarchical Task-network Planning},'' in \emph{International Conference on Artificial Intelligence Planning Systems}, 1994.

\bibitem{SchreiberLilotane2021}
D.~Schreiber, ``{Lilotane: A Lifted SAT-Based Approach to Hierarchical Planning},'' \emph{J. Artif. Int. Res.}, vol.~70, p. 1117–1181, may 2021.

\bibitem{nau2021gtpyhop}
D.~Nau, Y.~Bansod, S.~Patra, M.~Roberts, and R.~Li, ``{GTPyhop: A hierarchical goal+ task planner implemented in Python},'' \emph{HPlan 2021}, p.~21, 2021.

\bibitem{LearningHTNbyObservationNejati2006}
N.~Nejati, P.~Langley, and T.~Konik, ``Learning hierarchical task networks by observation,'' in \emph{Proceedings of the 23rd International Conference on Machine Learning}, ser. ICML '06.\hskip 1em plus 0.5em minus 0.4em\relax New York, NY, USA: Association for Computing Machinery, 2006, p. 665–672.

\bibitem{pmlr-v155-chen21d}
K.~Chen, N.~S. Srikanth, D.~Kent, H.~Ravichandar, and S.~Chernova, ``Learning hierarchical task networks with preferences from unannotated demonstrations,'' in \emph{Proceedings of the 2020 Conference on Robot Learning}, ser. Proceedings of Machine Learning Research, J.~Kober, F.~Ramos, and C.~Tomlin, Eds., vol. 155.\hskip 1em plus 0.5em minus 0.4em\relax PMLR, 16--18 Nov 2021, pp. 1572--1581.

\bibitem{zhuoLearningHierarchicalTask2014}
H.~H. Zhuo, H.~Munoz-Avila, and Q.~Yang, ``Learning hierarchical task network domains from partially observed plan traces,'' \emph{Artificial intelligence}, vol. 212, pp. 134--157, 2014.

\bibitem{hayesAutonomouslyConstructingHierarchical2016}
B.~Hayes and B.~Scassellati, ``Autonomously constructing hierarchical task networks for planning and human-robot collaboration,'' in \emph{2016 {{IEEE International Conference}} on {{Robotics}} and {{Automation}} ({{ICRA}})}.\hskip 1em plus 0.5em minus 0.4em\relax {IEEE Press}, pp. 5469--5476.

\bibitem{liLearningProbabilisticHierarchical2014a}
N.~Li, W.~Cushing, S.~Kambhampati, and S.~Yoon, ``Learning {{Probabilistic Hierarchical Task Networks}} as {{Probabilistic Context-Free Grammars}} to {{Capture User Preferences}},'' vol.~5, no.~2, pp. 29:1--29:32.

\bibitem{hafnerDeepHierarchicalPlanning2022}
D.~Hafner, K.-H. Lee, I.~Fischer, and P.~Abbeel, ``Deep hierarchical planning from pixels,'' \emph{Advances in Neural Information Processing Systems}, vol.~35, pp. 26\,091--26\,104, 2022.

\bibitem{Barto2003RecentAI}
A.~G. Barto and S.~Mahadevan, ``Recent advances in hierarchical reinforcement learning,'' \emph{Discrete Event Dynamic Systems}, vol.~13, pp. 41--77, 2003.

\bibitem{AlbusTheoryOfIntelligence1991}
J.~Albus, ``Outline for a theory of intelligence,'' \emph{IEEE Transactions on Systems, Man, and Cybernetics}, vol.~21, no.~3, pp. 473--509, 1991.

\bibitem{Albus4DRCS2002}
J.~S. Albus, H.-M. Huang, E.~R. Messina, K.~Murphy, M.~Juberts, A.~Lacaze, S.~B. Balakirsky, M.~O. Shneier, T.~H. Hong, H.~A. Scott \emph{et~al.}, ``{4D/RCS Version 2.0: A reference model architecture for unmanned vehicle systems},'' 2002.

\bibitem{AlamiArchitectureForAutonomy1998}
R.~Alami, R.~Chatila, S.~Fleury, M.~Ghallab, and F.~Ingrand, ``An architecture for autonomy,'' \emph{The International Journal of Robotics Research}, vol.~17, no.~4, pp. 315--337, 1998.

\bibitem{VolpeClaraty2001}
R.~Volpe, I.~Nesnas, T.~Estlin, D.~Mutz, R.~Petras, and H.~Das, ``{The CLARAty architecture for robotic autonomy},'' in \emph{2001 IEEE Aerospace Conference Proceedings}, vol.~1, 2001, pp. 1/121--1/132 vol.1.

\bibitem{SpecialIssueDarpaGrandChallenge2006}
K.~Iagnemma and M.~Buehler, ``{Editorial for Journal of Field Robotics—Special Issue on the DARPA Grand Challenge: Editorial},'' \emph{J. Robot. Syst.}, vol.~23, no.~9, p. 655–656, sep 2006.

\bibitem{PintoArchitecture2019}
A.~Pinto, ``An open and modular architecture for autonomous and intelligent systems,'' in \emph{2019 IEEE International Conference on Embedded Software and Systems (ICESS)}, 2019, pp. 1--8.

\bibitem{14-gu2023review}
S.~Gu, L.~Yang, Y.~Du, G.~Chen, F.~Walter, J.~Wang, Y.~Yang, and A.~Knoll, ``A review of safe reinforcement learning: Methods, theory and applications,'' 2023.

\bibitem{3-ivanov2021compositional}
\BIBentryALTinterwordspacing
R.~Ivanov, K.~Jothimurugan, S.~Hsu, S.~Vaidya, R.~Alur, and O.~Bastani, ``Compositional learning and verification of neural network controllers,'' \emph{ACM Trans. Embed. Comput. Syst.}, vol.~20, no.~5s, sep 2021. [Online]. Available: \url{https://doi.org/10.1145/3477023}
\BIBentrySTDinterwordspacing

\bibitem{11-srinivasan2020learning}
K.~Srinivasan, B.~Eysenbach, S.~Ha, J.~Tan, and C.~Finn, ``Learning to be safe: Deep rl with a safety critic,'' 2020.

\bibitem{12-thananjeyan2021recovery}
B.~Thananjeyan, A.~Balakrishna, S.~Nair, M.~Luo, K.~Srinivasan, M.~Hwang, J.~E. Gonzalez, J.~Ibarz, C.~Finn, and K.~Goldberg, ``Recovery rl: Safe reinforcement learning with learned recovery zones,'' 2021.

\bibitem{13-bharadhwaj2021conservative}
H.~Bharadhwaj, A.~Kumar, N.~Rhinehart, S.~Levine, F.~Shkurti, and A.~Garg, ``Conservative safety critics for exploration,'' 2021.

\bibitem{4-chow2018lyapunovbased}
Y.~Chow, O.~Nachum, E.~Duenez-Guzman, and M.~Ghavamzadeh, ``A lyapunov-based approach to safe reinforcement learning,'' 2018.

\bibitem{5-chow2019lyapunovbased}
Y.~Chow, O.~Nachum, A.~Faust, E.~Duenez-Guzman, and M.~Ghavamzadeh, ``Lyapunov-based safe policy optimization for continuous control,'' 2019.

\bibitem{6-salamati2020lyapunov}
M.~Salamati, S.~Soudjani, and R.~Majumdar, ``A lyapunov approach for time bounded reachability of ctmcs and ctmdps,'' 2020.

\bibitem{8-cheng2019endtoend}
R.~Cheng, G.~Orosz, R.~M. Murray, and J.~W. Burdick, ``End-to-end safe reinforcement learning through barrier functions for safety-critical continuous control tasks,'' 2019.

\bibitem{9-li2019temporal}
X.~Li and C.~Belta, ``Temporal logic guided safe reinforcement learning using control barrier functions,'' 2019.

\bibitem{10-emam2022safe}
Y.~Emam, G.~Notomista, P.~Glotfelter, Z.~Kira, and M.~Egerstedt, ``Safe reinforcement learning using robust control barrier functions,'' 2022.

\bibitem{23-s2rl2023}
B.~Gangopadhyay, P.~Dasgupta, and S.~Dey, ``Safe and stable rl (s2rl) driving policies using control barrier and control lyapunov functions,'' \emph{IEEE Transactions on Intelligent Vehicles}, vol.~8, no.~2, pp. 1889--1899, 2023.

\bibitem{19-lütjens2020certified}
B.~Lütjens, M.~Everett, and J.~P. How, ``Certified adversarial robustness for deep reinforcement learning,'' 2020.

\bibitem{24-pinto2017robust}
L.~Pinto, J.~Davidson, R.~Sukthankar, and A.~Gupta, ``Robust adversarial reinforcement learning,'' 2017.

\bibitem{20-sadeghi2017cad2rl}
F.~Sadeghi and S.~Levine, ``Cad2rl: Real single-image flight without a single real image,'' 2017.

\bibitem{hock-auro19}
A.~Hock and A.~P. Schoellig, ``Distributed iterative learning control for multi-agent systems,'' \emph{{Autonomous Robots}}, vol.~43, no.~8, pp. 1989--2010, 2019.

\bibitem{8039204}
U.~Rosolia and F.~Borrelli, ``{Learning Model Predictive Control for Iterative Tasks. A Data-Driven Control Framework},'' \emph{IEEE Transactions on Automatic Control}, vol.~63, no.~7, pp. 1883--1896, 2018.

\bibitem{rosolia2019learning}
------, ``Learning how to autonomously race a car: a predictive control approach,'' 2019.

\bibitem{chaal2020}
M.~Chaal, O.~Valdez~Banda, J.~Glomsrud, S.~Basnet, S.~Hirdaris, and P.~Kujala, ``{A framework to model the STPA hierarchical control structure of an autonomous ship},'' \emph{Safety Science}, vol. 132, p. 104939, 12 2020.

\bibitem{dinh2020}
N.~Dinh, M.~Sualeh, D.~Kim, and G.-W. Kim, ``A hierarchical control system for autonomous driving towards urban challenges,'' \emph{Applied Sciences}, vol.~10, p.~26, 05 2020.

\bibitem{DIXIT201876}
S.~Dixit, S.~Fallah, U.~Montanaro, M.~Dianati, A.~Stevens, F.~Mccullough, and A.~Mouzakitis, ``Trajectory planning and tracking for autonomous overtaking: State-of-the-art and future prospects,'' \emph{Annual Reviews in Control}, vol.~45, pp. 76--86, 2018.

\bibitem{MONTOYA201787}
A.~Montoya, C.~Guéret, J.~E. Mendoza, and J.~G. Villegas, ``The electric vehicle routing problem with nonlinear charging function,'' \emph{Transportation Research Part B: Methodological}, vol. 103, pp. 87--110, 2017, green Urban Transportation.

\bibitem{vallon2024learning2}
C.~Vallon, A.~Pinto, B.~Stellato, and F.~Borrelli, ``{Learning Hierarchical Control For Multi-Agent Capacity-Constrained Systems},'' \emph{To appear in 63rd IEEE Conference on Decision and Control}, 2024.

\bibitem{Monimoy}
M.~Bujarbaruah, ``Robust model predictive control with data-driven learning,'' Ph.D. dissertation, UC Berkeley, 2022.

\bibitem{ClusterRECVRP}
M.~Pustilnik and F.~Borrelli, ``{Clustering Heuristics for Robust Energy Capacitated Vehicle Routing Problem (ECVRP)},'' \emph{arXiv preprint arXiv:2403.13906}, 2024.

\end{thebibliography}

\end{document}